\documentclass[12pt, a4paper]{article}
\usepackage{graphicx}
\usepackage{fullpage}
\usepackage{amsfonts, amsthm}
\usepackage{amsmath}

\newtheorem{theorem}{Theorem}
\newtheorem{lemma}{Lemma}

\newcommand{\per}{\mbox{\rm Per}}
\renewcommand{\det}{\mbox{\rm det}}
\newcommand{\ignore}[1]{}
\newcommand{\FF}{F}
\newcommand{\bfy}{y}
\newcommand{\bfx}{x}
\newcommand{\bfz}{z}
\newcommand{\labell}{\label}
\newcommand{\bfi}{i}

\newcommand{\singlespacing}%
{\small\normalsize}
\newcommand{\bfone}{{\bf 1}}

\newlength{\boxwidth}
\newcommand{\boxer}[1]{\begin{center}
                        \fbox{\fbox{
                                \begin{minipage}{\boxwidth}
                                \vspace*{.2in}
\label{style}

                                #1
                                \vspace*{.2in}
                                \end{minipage}}}
                        \end{center}}

\title{A Simple Algorithm for Hamiltonicity}

\author{Hasan Abasi \hspace{1.0in}  Nader H. Bshouty\\
Department of Computer Science\\
 Technion, 32000}
\date{\today}
\begin{document}
\maketitle
\bibliographystyle{plain}

\begin{abstract}
We develop a new algebraic technique that solves the following
problem: Given a black box that contains an arithmetic circuit $f$
over a field of characteristic $2$ of degree~$d$. Decide whether
$f$, expressed as an equivalent multivariate polynomial, contains
a multilinear monomial of degree $d$.

This problem was solved by Williams \cite{W} and Bj\"orklund et. al. \cite{BHKK}
for a white box (the circuit is given as an input)
that contains arithmetic circuit. We show a simple black box
algorithm that solves the problem with the same time complexity.

This gives a simple randomized algorithm for the simple $k$-path
problem for directed graphs of the same time
complexity\footnote{$O^*(f(k))$ is $O(poly(n)\cdot f(k))$}
$O^*(2^k)$ as in \cite{W} and with reusing the same ideas from
\cite{BHKK} with the above gives another algorithm (probably not
simpler) for undirected graphs of the same time complexity
$O^*(1.657^k)$ as in \cite{B10,BHKK}.
\end{abstract}

\section{Introduction}
Given a graph $G$ on $n$ vertices, the $k$-{\it path problem} asks
whether $G$ contains a simple path of length $k$. For $k=n$ the
problem is the Hamiltonian path problem in graphs.

Williams proved in \cite{W}
\begin{theorem} The directed $k$-path problem can be solved in time
$O^*(2^k)$ by a randomized algorithm with constant one sided
error.
\end{theorem}

Bj\"orklund et. al. proved in \cite{B10,BHKK}

\begin{theorem}
The undirected $k$-path problem can be solved in time
$O^*(1.657^k)$ by a randomized algorithm with constant one sided
error.
\end{theorem}

Both algorithms are based on using a dynamic programming for
constructing an arithmetic circuit $f_G$ over a field of
characteristic $2$ where $f_G\not\equiv 0$ if and only if there is
a $k$-path in the graph. The constructions are in two stages. In
the first stage the algorithm construct an arithmetic circuit that
is equivalent to a multivariate polynomial that contains a
monomial $x_{i_1}\cdots x_{i_k}$ for every $k$-path $v_{i_1} \to
v_{i_2}\to \cdots \to v_{i_k}$ in the graph. In the second stage
the algorithm constructs a modified circuit where all
non-multilinear monomials are removed. Our algorithm takes the
first construction and then gives a simple black box test that
tests whether the multivariate polynomial contains a multilinear
monomial.

This gives a simple
randomized algorithm for the simple $k$-path problem for
directed graphs of the same complexity $O^*(2^k)$ as in Theorem~1
and with reusing the same ideas from \cite{BHKK} with the above
gives another solution (probably not simpler)
for undirected graphs of the same complexity
$O^*(1.657^k)$ as in Theorem~2.

\section{Main Result}
In this section we prove our main result.

Let $x=(x_1,\ldots,x_n)$ and $y=(y_1,\ldots,y_m)$ be two sets of
variables. A monomial over $x$ of degree $k$ is
$M:=x_{i_1}x_{i_2}\cdots x_{i_k}$ where $1\le i_1\le \cdots\le
i_k\le n$. When $i_1,i_2,\ldots,i_k$ are distinct then we call $M$
multilinear monomial. A monomial over $x$ and $y$ is $M:=M_1M_2$
where $M_1$ is a monomial over $x$ and $M_2$ is a monomial over
$y$. We say that $M$ is multilinear in $x$ if $M_1$ is
multilinear. Every arithmetic circuit over the field $F$ with the
inputs $x$ and $y$ can be expressed as a multivariate polynomial
$f$ in $F[x,y]$. The degree of $f$ in $x$ is the degree of $f$ in
$F[y][x]$. I.e. the degree $f$ when it is expressed as a
multivariate polynomial in $x$ with coefficients from $F[y]$.

Our main result is the following

\begin{theorem}\label{TT} Let $x$ and $y$ be two sets of variables.
Given a black box that contains an arithmetic circuit for $f\in
F[x,y]$ over a field $F$ of characteristic $2$ of degree $k$ in
$x$ and $d$ in $y$. There is a randomized algorithm with constant
one sided error, that runs in $O^*(poly(d)\cdot 2^k)$ time, asks
$2^k$ substitution queries and decides whether $f$ contains a
multilinear monomial in $x$ of degree $k$.
\end{theorem}

Before we give the proof of the above theorem we introduce

\subsection{The Operator $\phi_k$} In this subsection we
introduce a notion from \cite{B}.

Let $\FF$ be any field of characteristic $2$. Consider a
multivariate polynomial $f\in\FF[x_1,\ldots,x_n]$ of degree~$k$.
Let $\bfz=(\bfz_1,\ldots,\bfz_k)$ where
$\bfz_{i}=(z_{i,1},\ldots,z_{i,n})$ are new indeterminates for
$i=1,\ldots,k$. Define the operator {$\phi_k:\FF[\bfx]\to
\FF[\bfz]$}{}
\begin{eqnarray}\labell{pdf}
\phi_k f = \sum_{J\subseteq [k]} f\left(\sum_{i\in J}
\bfz_i\right),
\end{eqnarray}
where $\sum_{i\in \emptyset}\bfz_i=0$. In \cite{B} Bshouty showed

\begin{lemma} \labell{Ry}
We have
\begin{enumerate}
\item \labell{Ry1}For a monomial $M$ that is non-multilinear of
degree $k$ we have $\phi_k M\equiv 0$.

\item \labell{Ry15}For a monomial $M$ of degree less than $k$ we
have $\phi_k M\equiv 0$.

\item \labell{Ry2}For multilinear monomial $M_\bfi=x_{i_1}\cdots
x_{i_k}$ of degree $k$ we have $\phi_kM_\bfi=\det\
Z_{M_\bfi}\not\equiv 0$ where
$$Z_{M_\bfi}(\bfz) =\left(\begin{array}{cccc}
z_{1,i_1}&z_{1,i_2}&\cdots&z_{1,i_k}\\
z_{2,i_1}&z_{2,i_2}&\cdots&z_{2,i_k}\\
\vdots& \vdots & \vdots & \vdots \\
z_{k,i_1}&z_{k,i_2}&\cdots&z_{k,i_k}\end{array}\right).$$
\end{enumerate}
\end{lemma}

Suppose
$$f(\bfx)=\sum_{\bfi\in I}\lambda_\bfi x_{i_1}\cdots
x_{i_k}+g(\bfx)$$ where $\bfx=(x_1,\ldots,x_n)$, $g(\bfx)$
contains non-multilinear monomials of degree $k$ and monomials of
degree less than $k$, $\bfi=(i_1,i_2,\ldots,i_k)$, $1\le  i_1<
i_2< \cdots< i_k\le n$, $\lambda_\bfi\not=0$ and $I\subset [n]^k$.
Since $\phi_k$ is linear we have
\begin{eqnarray}\labell{bottle1}
(\phi_k f)(\bfz_1,\ldots,\bfz_k)=\sum_{\bfi\in
I}\lambda_\bfi\cdot\det\left(Z_{M_\bfi}(\bfz_1,\ldots,\bfz_k)\right).
\end{eqnarray}
Notice that for two distinct $i^{(1)},i^{(2)}\in [n]^k$, the
monomials of $\det(Z_{M_{i^{(1)}}})$ and $\det(Z_{M_{i^{(2)}}})$
are disjoint. Therefore if $\deg(f)\le k$ then $\phi_k f\not\equiv
0$ if and only if $f$ contains a multilinear monomial of degree
$k$.

\subsection{Proof of Theorem~\ref{TT}}
Let $f\in F[x,y]$ where $F$ is a field of characteristic $2$.
Suppose
$$f(\bfx,y)=\sum_{\bfi\in I}\lambda_\bfi(y) x_{i_1}\cdots
x_{i_k}+g(\bfx,y)$$ is a multivariate polynomial of degree $k$ in
$x$ and $d$ in $y$ and $g(\bfx,y)$ is a multivariate polynomial
that contains monomials that are not multilinear in $x$ of degree
$k$, $\bfi=(i_1,i_2,\ldots,i_k)$, $1\le i_1< i_2< \cdots< i_k\le
n$, $\lambda_\bfi(y)\not\equiv 0$ and $I\subset [n]^k$. Then
\begin{eqnarray}\labell{bottle2}
(\phi_k f)(\bfz_1,\ldots,\bfz_k,y)=\sum_{\bfi\in
I}\lambda_\bfi(y)\det\left(Z_{M_\bfi}(\bfz_1,\ldots,\bfz_k)\right).
\end{eqnarray}
Therefore, $f(x,y)$ contains a multilinear monomial in $x$ if and
only if $\phi_kf\not\equiv 0$. Now since the degree of
$\phi_kf(x,y)$ is at most $d+k$, by Schwartz-Zippel zero testing
and since each substitution in $\phi_kf(x,y)$ can be simulated by
$2^k$ substitutions in $f(x,y)$ the result follows.

\section{William's Result}
The purpose of this section is to reduce the question of whether a
directed graph $G$ contains a simple $k$-path, to that of whether
a certain multivariate polynomial $f(x,y)$ contains a multilinear
monomial in $x$. We now describe this reduction.

Let $G(V,E)$ be a directed graph where $V=\{1,2,\ldots,n\}$. Let
$A$ be the adjacency matrix. Let $B^{(m)}$ be an $n\times n$
matrices, $m=2,\ldots,k$, such that $B^{(m)}_{i,j}=x_i\cdot
y_{m,i}\cdot A_{i,j}$ where $x_i$ and $y_{m,i}$ are
indeterminates. Let, $\bfy=(\bfy_1,\ldots,\bfy_{k})$ and $\bfy_m
=(y_{m,1},\ldots,y_{m,n})$. Let
$\bfx\ast\bfy=(x_1y_{1,1},\ldots,x_ny_{1,n})$. Consider the
polynomial $P_G(\bfx,\bfy)=\bfone B^{(k)}B^{(k-1)}\cdots
B^{(2)}(\bfx\ast\bfy)$. It is easy to see that

$$P_G(\bfx,\bfy)=\sum_{i_1\to i_2\to \cdots \to i_{k} \in G} x_{i_1}\cdots x_{i_{k}} y_{1,i_1}\cdots y_{k,i_{k}}$$
Obviously, no two paths have the same monomial in $P_G$.
Therefore, for any field, there is a simple $k$-path if and only
if $P_G(\bfx,\bfy)$ contains an multilinear monomial of degree
$k$. Now the result follows from Theorem~\ref{TT}

The algorithm is in Figure~1.

\begin{figure}
\singlespacing{ \boxer{
\begin{small}
\noindent{\bf Algorithm Direct Hamiltonian($G(V,E)$,$k$).}
\begin{tabbing}\label{fig}
Build the circuit $P_G(\bfx,\bfy)=\bfone B^{(k)}B^{(k-1)}\cdots
B^{(2)}(\bfx\ast\bfy)$\\
\ Test if $\phi_k(P_G(\bfx,\bfy))=\sum_{J\subseteq [k]}P_G(x,
\sum_{i\in J} z_i)\equiv 0$
using Schwartz-Zippel lemma.\\
\ If $\phi_k(P_G(\bfx,\bfy))\not\equiv 0$ answer ``YES'' and halt.\\
\ Answer ``NO''
\end{tabbing}

\end{small} } }
\caption{\sl An algorithm for simple $k$-path in undirected
graph.}
\end{figure}

\section{Bj\"orklund et. al. Result}
In this section we give Bj\"orklund et. al. \cite{BHKK} result.

\subsection{Preliminary Results}
Let $G(V,E)$ be an undirected graph with $n=|V|$ vertices. A
$k$-{\it path} is $v_0,v_1,\ldots,v_k$ such that
$\{v_i,v_{i+1}\}\in E$ for every $i=0,\ldots,k-1$. A $k$-path
$v_0,v_1,\ldots,v_k$ is called {\it simple} if the vertices in the
path are distinct. Notice here that unlike the previous definition
the length $k$ of the path is the number of edges (which is the
number of vertices$-1$) and not the number of vertices.

Let $V=V_1\cup V_2$ be a partition of $V$.  Let $E_1=E(V_1)$ and
$E_2=E(V_2)$ be the set of edges with both ends in $V_1$ and
$V_2$, respectively. Our goal is to find a simple $k$-path that
starts from some fixed vertex. For a path $p=v_0,v_1,\ldots, v_k$
we define the multiset of vertices in $p$ as
$V(p)=\{v_0,\ldots,v_k\}$ and the (undirected) edges in $p$ as the
multiset $E(p)=\{\{v_0,v_1\},\ldots,\{v_{k-1},v_k\}\}$. When we
write $V(p)\cap V_1$ (or $E(p)\cap E_2$) we mean the multiset that
contains the elements in $V(p)$ that are also in $V_1$.

Define for every edge $e\in E$ a variable $x_e$, for every vertex
$v\in V_1$ a variable $y_v$ and for every edge $e\in E_2$ a
variable $z_e$. Let $x=(x_e)_{e\in E}$ $y=(y_v)_{v\in V_1}$ and
$z=(z_e)_{e\in E_2}$. For every $k$-path $p=v_0,v_1,\ldots, v_k$
we define a monomial over any field of characteristic~$2$
$M_p=X_pY_pZ_p$ where $$X_p=\left(\prod_{e\in E(p)} x_e\right) ,
Y_p=\left(\prod_{v\in V(p)\cap V_1}y_v\right) \mbox{\ and\ }
Z_p=\left(\prod_{e\in E(p)\cap E_2} z_e\right).$$

Note here that if $e$ appears twice in $E(p)$ then $x_e$ appears
twice in $X_p$.

A path $p=v_0,v_1,\ldots,v_k$ is called {\it $(r,s)$-legitimate
$k$-path} with the partition $V=V_1\cup V_2$ if $|V(p)\cap
V_1|=r$, $|E(p)\cap E_2|=s$ and it contains no three consecutive
vertices $v_i,v_{i+1},v_{i+2}$ where $v_{i+2}=v_i$, $v_i\in V_2$
and $v_{i+1}\in V_1$. Fix a vertex $v_0\in V_1$. We denote by
${\cal L}_{k,r,s}(v_0,V_1,V_2)$ the set of all $(r,s)$-legitimate
$k$-paths in $G$ with the partition $V_1\cup V_2=V$ that starts
from $v_0\in V_1$. Define
$$F^{v_0,V_1,V_2}_{k,r,s}(x,y,z)=\sum_{p\in {\cal L}_{k,r,s}(v_0,V_1,V_2)}
M_p.$$

We now prove the following results.
\begin{lemma}\label{L1} Given an undirected graph $G=(V,E)$, a partition $V=V_1\cup
V_2$, $v_0\in V_1$ and two integers $s$ and $r$. There is a
deterministic polynomial time algorithm that construct a
polynomial size circuit for the function
$F^{v_0,V_1,V_2}_{k,r,s}(x,y,z)$.
\end{lemma}

The following lemma follows from \cite{BHKK}. We give here the proof for completeness
\begin{lemma}\label{L2} $M_p=X_pY_pZ_p$ is a monomial in
$F^{v_0,V_1,V_2}_{k,r,s}(x,y,z)$ and $Y_pZ_p$ is multilinear
 if and only if $p$ is a $(r,s)$-legitimate simple $k$-path with
 the partition $V_1\cup V_2=V$ that starts from $v_0$.
\end{lemma}

The following follows immediately from Theorem~\ref{TT}.

\begin{lemma}\label{L3} There is a randomized algorithm with constant, one sided
error, that runs in time $O^*(2^{r+s})$ for the following decision
problem: Given a black box for the multivariate polynomial
$f(x,y,z):=F^{v_0,V_1,V_2}_{k,r,s}(x,y,z)$ over a field of
characteristic $2$, decides whether $f$ contains a monomial
$M_p=X_pY_pZ_p$ where $Y_pZ_p$ is multilinear.
\end{lemma}

\noindent {\bf Proof of Lemma~\ref{L1}}. For any two vertices
$u_1,u_2\in V$ we define ${\cal L}_{k,r,s}(v_0,V_1,V_2,u_1,u_2)$
the set of all $(r,s)$-legitimate $k$-paths that start with $v_0$
and end with $u_1,u_2$. I.e, $u_2$ is the last node in the path
and $u_1$ proceeds it. Define
$$F^{v_0,V_1,V_2,u_1,u_2}_{k,r,s}=\sum_{p\in {\cal
L}_{k,r,s}(v_0,V_1,V_2,u_1,u_2)} M_p.$$ Then
$$F^{v_0,V_1,V_2}_{k,r,s}(x,y,z)=\sum_{u_1,u_2\in V}F^{v_0,V_1,V_2,u_1,u_2}_{k,r,s}.$$

We now show, using dynamic programming, that
$F^{v_0,V_1,V_2,u_1,u_2}_{k,r,s}$ can be computed in polynomial
time. For a vertex $v$ let $N(v)$ be the neighbor vertices of $v$.
For a predicate $A$ we define $[A]=1$ if $A$ is true and $0$
otherwise. Now it is easy to verify the following recurrence
formula
\begin{enumerate}
\item If $k\ge 2$, $r\le k+1,s\le k$ and $\{u_1,u_2\}\in E$ then
\begin{eqnarray*}
F^{v_0,V_1,V_2,u_1,u_2}_{k,r,s}&=& [u_2\in V_1]\cdot
x_{\{u_1,u_2\}}y_{u_2} \sum_{w\in
N(u_1)}F^{v_0,V_1,V_2,w,u_1}_{k-1,r-1,s}\\
&&+[u_2\in V_2\wedge u_1\in V_2 ]\cdot
x_{\{u_1,u_2\}}z_{\{u_1,u_2\}} \sum_{w\in
N(u_1)}F^{v_0,V_1,V_2,w,u_1}_{k-1,r,s-1}\\
&&+[ u_2\in V_2\wedge u_1\in V_1]\cdot x_{\{u_1,u_2\}} \sum_{w\in
N(u_1)\backslash \{u_2\}}F^{v_0,V_1,V_2,w,u_1}_{k-1,r,s}\\
\end{eqnarray*}
\ignore{If $\{u_1,u_2\}\not\in E$ then
$F^{v_0,V_1,V_2,u_1,u_2}_{k,r,s}\equiv 0$.}

\item If $k=1$, $u_1=v_0$, $u_2\in V_2$, $r=1$ and $s=0$ then
$F^{v_0,V_1,V_2,u_1,u_2}_{k,r,s}=x_{\{v_0,u_2\}}y_{v_0}$.

\item If $k=1$, $u_1=v_0$, $u_2\in V_1$, $r=2$ and $s=0$ then
$F^{v_0,V_1,V_2,u_1,u_2}_{k,r,s}=x_{\{v_0,u_2\}}y_{v_0}y_{u_2}$.

\item Otherwise $F^{v_0,V_1,V_2,u_1,u_2}_{k,r,s}=0$.
\end{enumerate}

Since $u_1,u_2,k,r,s$ can take at most $k^2(k+1)n^2$ different
values the above recurrence can be computed in polynomial
time.\qed
\begin{eqnarray*}
\end{eqnarray*}
 \noindent {\bf Proof of
Lemma~\ref{L2}}. ($\Leftarrow$) Let $p=v_0,v_1,\ldots,v_k$ be any
$(r,s)$-legitimate simple $k$-path with the partition $V_1\cup
V_2=V$. Then $p\in {\cal L}_{k,r,s}(v_0,V_1,V_2)$. Since $p$ is
simple $Y_pZ_p$ is multilinear. We now need to show that no other
path $p'$ satisfies $M_{p'}=M_p$. If $M_p=M_{p'}$ then
$X_p=X_{p'}$ and since $p$ is simple and starts from $v_0$ by
induction on the path, $p\equiv p'$. Therefore $M_p$ is a
multilinear monomial in $F^{v_0,V_1,V_2}_{k,r,s}(x,y,z)$

($\Rightarrow$) We now show that all the monomials that correspond
to $(r,s)$-legitimate non-simple $k$-path $p=v_0,v_1,\ldots,v_k$
with the partition $V_1\cup V_2=V$ either vanish (because the
field is of characteristic~$2$) or are not multilinear.

Consider a $(r,s)$-legitimate non-simple $k$-path with the
partition $V_1\cup V_2=V$. Consider the first circuit $C$ in this
path. If $C=v_i,v_{i+1},v_i$ then either $v_i\in V_1$ and then
$Y_p$ contains $y_{v_i}^2$ or $v_i,v_{i+1}\in V_2$ and then $Z_p$
contains $z_{\{v_i,v_{i+1}\}}^2$. Notice that $p$ is legitimate
and therefore the case $v_i\in V_2$ and $v_{i+1}\in V_1$ cannot
happen.

Now suppose $|C|>2$, $C=v_i,v_{i+1},\ldots,v_{j},
v_{j+1}(=v_{i})$. Define $p_1=v_1,\ldots,v_{i-1}$ and
$p_2=v_{j+2},\ldots,v_k$. Then $p=p_1Cp_2$. If $v_i\in V_1$ then
$Y_p$ contains $y_{v_i}y_{v_{j+1}}=y_{v_i}^2$. Therefore we may
assume that $v_i\in V_2$. Define the path
$$\rho(p):=p_1C'p_2=\underline{v_0,v_1,\ldots,v_{i-1}},
\underline{v_i,v_{j},v_{j-1},\ldots,v_{i+1},v_{i}},
\underline{v_{j+2},v_{j+3},\ldots,v_{k}}.$$ We now show that
\begin{enumerate}
\item $\rho(\rho(p))=p$.

\item $\rho(p)$ is $(r,s)$-legitimate non-simple $k$-path with the
partition $V_1\cup V_2=V$ that starts with $v_0$.

\item $\rho(p)\not=p$ and $M_p=M_{\rho(p)}$.
\end{enumerate}
This implies that $M_p$ vanishes from
$F^{v_0,V_1,V_2}_{k,r,s}(x,y,z)$ because the characteristic of the
field is $2$. Let $p'=\rho(p)$. Since $C$ is the first circuit in
$p$ we have $v_0,v_1,\ldots,v_j$ are distinct and $v_{j+1}=v_i$.
This implies that $C'$ is the first circuit in $p'$ and therefore
$\rho(\rho(p))=\rho(p')=p$. This implies 1.

Obviously, $|V(p')\cap V_1|=|V(p)\cap V_1|=r$ and $|E(p')\cap
E_2|=|E(p)\cap E_2|=s$. Suppose $p'$ contains three consecutive
vertices $u,w,u$ such that $w\in V_1$ and $u\in V_2$. Then since
$p$ is $(r,s)$-legitimate path and $v_i\in V_2$ we have three
cases

\noindent {\bf Case I.} $u,w,u$ is in $C'$. Then $u,w,u$ is in $C$
which contradict the fact that $p$ is legitimate.

\noindent {\bf Case II.} $u=v_{i-2}\in V_2,w=v_{i-1}\in V_1$ and
$v_{i-2}=v_i$. In this case $C''=v_{i-2},v_{i-1},v_{i}$ is a
circuit in $p$ and then $p$ is not legitimate. A contradiction.

\noindent {\bf Case III.} $u=v_{i}\in V_2, w=v_{j+2}\in V_1$ and
$v_{j+3}=v_{i}$. In this case, $C''=v_{i}v_{j+2}v_{j+3}$ is a
circuit in $p$ and then $p$ is not legitimate. A contradiction.

This proves 2.

If $p= p'$ then $C=C'$ and since $|C|>2$ we get a contradiction.
Therefore $p\not= p'$. Since $E(p)=E(p')$, $V(p')\cap V_1=V(p)\cap
V_1$ and $E(p')\cap E_2=E(p)\cap E_2$ we also have $M_p=M_{p'}$.
This proves 3.\qed

\subsection{The Algorithm}

The following lemma is proved in \cite{BHKK} we give its proof for
completeness.

\begin{lemma}\label{prob}
Let $p=v_0,v_1,\ldots,v_k$ be a simple path. For a partition
$V_1,V_2$ selected uniformly at random where $v_0\in V_1$,
$$\Pr_{V_1,V_2}\left(|V(p)\cap V_1|=r,|E(p)\cap
E_2|=s\right)=2^{-k}{r \choose k-r-s+1}{k-r \choose s}.$$
\end{lemma}
\begin{proof} We will count the number of partitions $V_1,V_2$
that satisfies $|V(p)\cap V_1|=r$, $|E(p)\cap E_2|=s$ and $v_0\in
V_1$. Obviously, the probability in the lemma is $2^{-k}$ times
the number of such partitions.

Let $V_1,V_2$ be a partition such that $|V(p)\cap V_1|=r$,
$|E(p)\cap E_2|=s$ and $v_0\in V_1$.  Let
$v_{i_1}=v_0,v_{i_2},\ldots,v_{i_r}$ be the nodes in $V_1$. Let
$\bar s_j\ge 0$, $j=1,\ldots,r-1$ be the number of nodes in $V_2$
that are between $v_{i_j}$ and $v_{i_{j+1}}$. Let $\bar s_r$ be
the number of nodes in $V_2$ that are after $v_{i_r}$. Let $t$ be
the number of $\bar s_i$ that are not zero. For $j<r$ the number
of edges in $E_2$ that are between $v_{i_j}$ and $v_{i_{j+1}}$ is
$ s_i:=\max(\bar s_i-1,0)$. The number of edges in $E_2$ that are
after $v_{i_r}$ is $s_r:=\max(\bar s_r-1,0)$. Therefore
\begin{eqnarray}\label{kaer}
\sum_{i=1}^r\bar s_i=\sum_{i=1}^r s_i+t=s+t.\end{eqnarray}

Since the number of nodes in the path is
\begin{eqnarray}\label{kaer2}
k+1=r+\sum_{i=1}^r \bar s_i=r+s+t\end{eqnarray} we must have
$t=(k+1)-(r+s)$.

Now any partition that satisfies $\sum_{i=1}^r \bar s_i=s+t$ and
$t=(k+1)-(r+s)$ must also satisfy $|V(p)\cap V_1|=r$ and
$|E(p)\cap E_2|=s$. Therefore the number of such partitions is
equal to the number ways of writing $s+t$ as $\bar s_1+\bar
s_2+\cdots+\bar s_r$ where exactly $t$ of them are not zero. We
first select those $\bar s_{j_1},\ldots,\bar s_{j_t}$ that are not
zero. This can be done in ${r\choose t}$ ways. Then the number of
ways of writing $s+t$ as $\bar s_{j_1}+\cdots+\bar s_{j_t}$ where
$\bar s_{j_i}\ge 1$ is equal to the number of ways of writing $s$
as $x_1+\cdots+x_t$ where $x_i\ge 0$. The later is equal to
${t+s-1\choose t-1}.$ Therefore the number of such partitions is
$${t+s-1\choose t-1}{r\choose t}={k-r\choose s}{r\choose k-r-s+1}.$$

\end{proof}

We now give the algorithm.

The algorithm is in Figure~2. In the algorithm we randomly
uniformly choose a partition $V=V_1\cup V_2$ where $v_0\in V_1$.
This is done $T$ times for each vertex $v_0\in V$. If
$p=v_0,v_1,\ldots,v_k$ is simple path then by Lemma~\ref{prob},
the probability that no partition satisfies $|V(p)\cap V_1|=r$ and
$|E(p)\cap E_2|=s$ is at most
$$\left(1-2^{-k}{r \choose k-r-s+1}{k-r \choose s}\right)^T\le \frac{1}{4}.$$
Then by Lemma~\ref{L1}, $f=F^{v_0,V_1,V_2}_{k,r,s}(x,y,z)$ can be
constructed in $poly(n)$ time. By Lemma~2 $f$ has a multilinear
monomial. Then, by Lemma~\ref{L3}, this can be tested with
probability at least $3/4$. Therefore, if there is a simple path
then the algorithm answer ``YES'' with probability at least $1/2$.
If there is no simple path then, by Lemma~\ref{L2}, for every
$v_0\in V$ and every partition $V_1\cup V_2$,
$f=F^{v_0,V_1,V_2}_{k,r,s}(x,y,z)$ has no multilinear monomial. By
Lemma~\ref{L3}, $\phi_{|S|+|R|}(f)\equiv 0$ and the answer is
``NO'' with probability $1$.

\begin{figure}
\singlespacing{ \boxer{
\begin{small}
\noindent{\bf Algorithm Hamiltonian($G(V,E)$,$k$,$r$,$s$).}
\begin{tabbing}\label{fig}
For every $v_0\in V$\\
\ \ \ For $i=1$ to $T:=2^{k+1}/\left({r \choose k-r-s+1}{k-r \choose s}\right)$\\
\ \ \ \ \ \ \ Choose a random uniform partition $V=V_1\cup V_2$ where $v_0\in V_1$\\
\ \ \ \ \ \ \ Build the circuit $f=F^{v_0,V_1,V_2}_{k,r,s}(x,y,z)$
using
Lemma~\ref{L1}\\
\ \ \ \ \ \ \ Test if $\phi_{|S|+|R|}(f)=\sum_{S\subseteq [s]}
\sum_{R\subseteq [r]}f \left(x, \sum_{i\in R} y^{(i)},\sum_{j\in
S}
z^{(j)}\right)\equiv 0$\\
\ \ \ \ \ \ \ \ \ \ \ \ \ \ \ \ \ \ \ \ \ \ \ \ \ \ \ \
\ \ \ \ \ \ \ \ \ \ \ \ \ \ \ \ \ \ \ \ \ \ \ \ \ \ \ \  using Schwartz-Zippel lemma.\\
\ \ \ \ \ \ \ If $\phi_{|S|+|R|}(f)\not\equiv 0$ answer ``YES'' and halt.\\
Answer ``NO''
\end{tabbing}

\end{small} } }
\caption{\sl An algorithm for simple $k$-path in undirected graph.}
\end{figure}

This proves the following
\begin{lemma} Let $G$ be undirected graph. Algorithm {\bf Hamiltonian} {\rm ($G(V,E)$,$k$,$r$,$s$)}
runs in time
\begin{eqnarray}\label{fc}
O\left(\frac{2^{r+s+k}\cdot poly(n)}{{r \choose k-r-s+1}{k-r
\choose s}}\right)
\end{eqnarray}
and satisfies the following.
If $G$ contains a simple path
of length $k$ then {\bf Hamiltonian} {\rm ($G(V,E)$,$k$,$r$,$s$)} answer ``YES''
with constant probability. If $G$ contains no simple path of length $k$ then
{\bf Hamiltonian}{\rm ($G(V,E)$,$k$,$r$,$s$)} answer ``NO'' with probability $1$.
\end{lemma}

Now to minimize (\ref{fc}) we choose $r=\lfloor 0.5\cdot k\rfloor$
and $s=\lfloor 0.208\cdot k\rfloor$ and get the result.

\end{document}